\documentclass[11pt]{article}
\usepackage{a4wide}
\usepackage{amsmath,amssymb,color,graphicx,tabularx,wrapfig,xspace}

\usepackage[small]{caption}

\expandafter\ifx\csname pdfoptionalwaysusepdfpagebox\endcsname\relax\else
\fi

\newcommand{\BeginMyItemize}{\begin{itemize}\setlength{\itemsep}{-\parskip}}
\newcommand{\EndMyItemize}{\end{itemize}}
\newcommand{\myitemize}[1]{\BeginMyItemize #1 \EndMyItemize}

\newcommand{\BeginMyEnumerate}{\begin{enumerate}\setlength{\itemsep}{-\parskip}}
\newcommand{\EndMyEnumerate}{\end{enumerate}}
\newcommand{\myenumerate}[1]{\BeginMyEnumerate #1 \EndMyEnumerate}

\renewcommand{\leq}{\leqslant}
\renewcommand{\geq}{\geqslant}

\newtheorem{defin}{Definition}
  
\newtheorem{theo}[defin]{Theorem}
  \newenvironment{theorem}{\begin{theo} \sl}{\end{theo}}
\newtheorem{lem}[defin]{Lemma}
  \newenvironment{lemma}{\begin{lem} \sl}{\end{lem}}
  
\newtheorem{propo}[defin]{Proposition}
  \newenvironment{proposition}{\begin{propo} \sl}{\end{propo}}
\newtheorem{coro}[defin]{Corollary}
  \newenvironment{corollary}{\begin{coro} \sl}{\end{coro}}
\newtheorem{obse}[defin]{Observation}
  
\newtheorem{rem}[defin]{Remark}
  \newenvironment{remark}{\begin{rem} \rm}{\end{rem}}
\newtheorem{myfact}[defin]{Fact}
  
\newenvironment{proof}{\emph{Proof.}}{\hfill $\Box$ \medskip\\}

\newenvironment{myquote}%
  {\list{}{\leftmargin=4mm\rightmargin=4mm}\item[]}%
  {\endlist}
\newenvironment{claiminproof}{\begin{myquote}\noindent\emph{Claim.}}{\end{myquote}}
\newenvironment{proofinproof}{\begin{myquote}\noindent\emph{Proof.}}{\hfill $\diamond$ \end{myquote}}

\newcommand{\Reals}{{\Bbb R}}

\DeclareMathOperator{\polylog}{polylog}

\newcommand{\pa}{\mathrm{parent}}        

\newcommand{\C}{\ensuremath{\mathcal{C}}}

\newcommand{\E}{\ensuremath{\mathcal{E}}}

\newcommand{\subdiv}{\ensuremath{\mathcal{S}}}
\newcommand{\T}{\ensuremath{\mathcal{T}}}

\newcommand{\V}{\ensuremath{\mathcal{V}}}

\newcommand{\opt}{\mbox{{\sc opt}}\xspace}
\newcommand{\smallopt}{\mathrm{opt}}
\newcommand{\Xopt}{X_{\smallopt}}

\newcommand{\ol}{\overline}
\newcommand{\eps}{\varepsilon}

\newcommand{\etal}{\emph{et al.}\xspace}

\newcommand{\proj}[1]{\overline{#1}}
\newcommand{\cut}{\Xi}

\title{Removing Depth-Order Cycles Among Triangles: \\ An Efficient Algorithm
       Generating Triangular Fragments}

\author{Mark de Berg\thanks{Department of Computer Science, TU Eindhoven,
            P.O.~Box 513, 5600 MB Eindhoven, the Netherlands.
            Email: {\tt m.t.d.berg@tue.nl}.
            MdB is supported by the Netherlands' Organisation for
            Scientific Research (NWO) under project no.~024.002.003.}}
\date{}

\begin{document}

\maketitle

\begin{abstract}
More than 25 years ago, inspired by applications in computer graphics, Chazelle~\etal (FOCS 1991)
studied the following question: Is it possible to cut any set of $n$ lines
or other objects in $\Reals^3$ into a
subquadratic number of fragments such that the resulting fragments admit a depth order?
They managed to prove an $O(n^{9/4})$ bound on the number of fragments, but only for
the very special case of bipartite weavings of lines. Since then only little progress
was made, until a recent breakthrough by Aronov and Sharir (STOC 2016) who
showed that $O(n^{3/2}\polylog n)$ fragments suffice for any set of lines.
In a follow-up paper Aronov, Miller and Sharir (SODA 2017) proved an $O(n^{3/2+\eps})$ bound
for triangles, but their method uses high-degree algebraic arcs to perform the cuts.
Hence, the resulting pieces have curved boundaries.
Moreover, their method uses polynomial partitions, for which
currently no algorithm is known.
Thus the most natural version of the problem is still wide open:
Is it possible to cut any collection of $n$ disjoint triangles in $\Reals^3$ into a subquadratic
number of triangular fragments that admit a depth order? And if so, can we
compute the cuts efficiently?

We answer this question by presenting an algorithm that cuts any set of $n$ disjoint
triangles in $\Reals^3$ into $O(n^{7/4}\polylog n)$ triangular fragments that
admit a depth order. The running time of our algorithm is $O(n^{3.69})$.
We also prove a refined bound that depends on the number, $K$, of intersections between the
projections of the triangle edges onto the $xy$-plane: we show that
$O(n^{1+\eps} + n^{1/4} K^{3/4}\polylog n)$ fragments suffice to obtain a depth order.
This result extends to $xy$-monotone surface patches bounded by a constant number
of bounded-degree algebraic arcs in general position, constituting the first
subquadratic bound for surface patches. Finally, as a byproduct of our approach
we obtain a faster algorithm to cut a set of lines into $O(n^{3/2}\polylog n)$
fragments that admit a depth order. Our algorithm for lines runs in $O(n^{5.38})$ time,
while the previous algorithm uses $O(n^{8.77})$ time.
\end{abstract}

\section{Introduction}
\label{se:intro}
Let $T$ and $T'$ be two disjoint triangles (or other objects) in $\Reals^3$.
We that $T$ is \emph{below} $T'$---or, equivalently, that $T'$ is
\emph{above} $T$---when there is a vertical line~$\ell$ intersecting both $T$ and $T'$
such that $\ell\cap T$ has smaller $z$-coordinate than $\ell\cap T'$. We denote this
relation by $T\prec T'$. Note that two triangles may be unrelated by the
$\prec$-relation, namely when their vertical projections onto the $xy$-plane are disjoint.
Now let $\T$ be a collection of $n$ disjoint triangles in $\Reals^3$.
A \emph{depth order} (for the vertical direction) on $\T$ is a total order on
$\T$ that is consistent with the
$\prec$-relation, that is, an ordering $T_1,\ldots,T_n$ of the triangles such
that $T_i\prec T_j$ implies $i<j$.

Depth orders play an important role in many applications. For example, the
Painter's Algorithm from computer graphics performs hidden-surface
removal by rendering the triangles forming the objects in a scene one by one,
drawing each triangle ``on top of'' the already drawn ones. To give the correct result the
Painter's Algorithm must handle the triangles in depth order with respect to the viewing direction.
Several object-space hidden-surface removal algorithms and ray-shooting data structures
need a depth order as well. Depth orders also play a role when one wants
to assemble a product by putting its constituent parts one by one into place
using vertical translations~\cite{wl-grma-94}.
The problem of computing a depth order for a given set of objects
has therefore received considerable
attention~\cite{aks-rsisf-95,b-rsdoh-93,bg-vrscd-08,bos-cvdo-94}.
However, a depth order does not always exist since
there can be \emph{cyclic overlap}, as illustrated in Fig.~\ref{fi:cyclic-overlap}(i).
In such cases the algorithms above simply report that no depth exists.
What we would then like to do is to cut the triangles into fragments
such that the resulting set of fragments is acyclic (that is, admits a depth order).
This gives rise to the following problem: How many fragments are needed in the worst case
to ensure that a depth order exists? And how efficiently can we compute
a set of cuts resulting in a small set of fragments admitting a depth order?
\begin{figure}[h]
  \centering
  \includegraphics{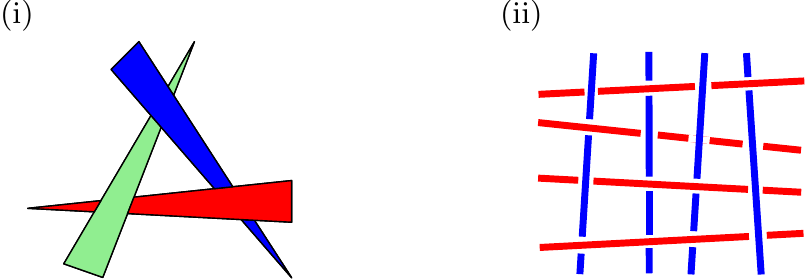}
  \caption{(i) Three triangles with cycle overlap.  (ii) A bipartite weaving.}
  \label{fi:cyclic-overlap}
\end{figure}

The problem of bounding the worst-case number of fragments needed to remove
all cycles from the depth-order relation has a long history.
In the special case of lines (or line segments)
one can easily get rid of all cycles using $O(n^2)$ cuts:
project the lines onto the $xy$-plane and cut each line at
its intersection points with the other lines.
A lower bound on the worst-case number of cuts is $\Omega(n^{3/2})$~\cite{cegpsss-ccclr-92}.
It turned out to be amazingly hard to get any subquadratic upper bound.
In 1991 Chazelle~\etal~\cite{cegpsss-ccclr-92} obtained such a bound, but
only for so-called bipartite weavings; see Fig.~\ref{fi:cyclic-overlap}(ii).
Moreover, their $O(n^{9/4})$ bound is
still far away from the $\Omega(n^{3/2})$ lower bound. Later Aronov~\etal~\cite{aks-ctcls-05}
obtained a subquadratic upper bound for general sets of lines, but they only get of all triangular
cycles---that is, cycles consisting of three lines---and their bounds are only slightly subquadratic:
they use $O(n^{2-1/69}\log^{16/69} n)$ cuts to remove all triangular cycles.
(They obtained a slightly better bound of $O(n^{2-1/34}\log^{8/17}n)$ for removing
all so-called elementary triangular cycles.) Finally, several authors studied the
algorithmic problem of computing a minimum-size complete cut set---a \emph{complete cut set}
is a set of cuts that removes all cycles from the depth-order relation---for a set of lines
(or line segments). Solan~\cite{s-ccoris-98} and Har-Peled and Sharir~\cite{hs-oplpa-01}
gave algorithms that produce a complete cut sets of size roughly $O(n\sqrt{\opt})$,
where $\opt$ is the minimum size of any complete cut set for the given lines.
Aronov~\etal~\cite{abgm-ccrs-08} showed that
this problem is {\sc np}-hard, and they presented an algorithm that computes
a complete cut set of size $O(\opt \cdot \log\opt\cdot \log\log\opt)$
in $O(n^{4+2\omega} \log^2n)=O(n^{8.764})$ time, where
$\omega<2.373$ is the exponent of the best matrix-multiplication algorithm.

Eliminating depth cycles from a set of triangles
is even harder than it is for lines. The trivial bound
on the number of fragments is $O(n^3)$, which can for instance be obtained by taking a vertical
cutting plane containing each triangle edge. Paterson and Yao~\cite{py-ebsph-90}
showed already in 1990 that any set of disjoint triangles admits a so-called
\emph{binary space partition} (BSP) of size $O(n^2)$, which immediately implies
an $O(n^2)$ bound on the number of fragments needed to remove all
cycles. Indeed, a BSP ensures that the resulting set of triangle fragments is acyclic
for any direction, not just for the vertical direction.
Better bounds on the size of BSPs are known for fat objects (or, more generally,
low-density sets)~\cite{b-lsbsp-00} and for axis-aligned
objects~\cite{agmv-bspfr-00,py-obspo-92,t-bspaafr-08}, but
for arbitrary triangles there is an $\Omega(n^2)$ lower bound on the worst-case
size of a BSP~\cite{c-cpp-86}. Thus to get a subquadratic bound on the number of fragments
needed to obtain a depth order, one needs a different approach.
\medskip

In 2016, using Guth's polynomial partitioning technique~\cite{g-ppsv-15},
Aronov and Sharir~\cite{as-atbedc-16} achieved a breakthrough in the area
by proving that any set of $n$ lines in $\Reals^3$ in general position
can be cut into $O(n^{3/2} \polylog n)$ fragments such that the resulting set
of fragments admits a depth order. A complete cut set of size $O(n^{3/2}\polylog n)$ can
then be computed using the algorithm of Aronov~\etal\cite{abgm-ccrs-08} mentioned above.
(They also gave a more refined bound for line segments, which depends on the number
of intersections, $K$, between the segments in the projection. More precisely, they
show that $O(n + n^{1/2} K^{1/2}\polylog n)$ cuts suffice.)
In a follow-up paper, Aronov, Miller and Sharir~\cite{ams-adcat-17}
extended the result to triangles: they show that, for any fixed $\eps>0$,
any set of disjoint triangles in general position can be cut into
$O(n^{3/2+\eps})$ fragments that admit a depth order. This may seem to almost settle the problem
for triangles, but the result of Aronov, Miller and Sharir has two serious drawbacks.
\myitemize{
\item The technique does not result in triangular
      fragments, since it cuts the triangles using algebraic arcs. The degree of these
      arcs is exponential in the parameter $\eps$ appearing in $O(n^{3/2+\eps})$ bound.
\item The technique, while being in principle constructive, does not give an efficient
      algorithm, since currently no algorithms are known for constructing
      Guth's polynomial partitions.
}
Arguably, the natural way to pose the problem for triangles is that one requires
the fragments to be triangular as well---polygonal fragments can always be decomposed
further into triangles, without increasing the number of fragments asymptotically---so
especially the first drawback is a major one. Indeed, Aronov, Miller and Sharir
state that ``\emph{It is a natural open problem to determine
whether a similar bound can be achieved with straight cuts [\ldots].
Even a weaker bound, as long as it is subquadratic and generally applicable, would be of
great significance.}'' Another open problem stated by Aronov, Miller and Sharir
is to extend the result to surface patches: ``\emph{Extending the technique to curved objects (e.g.,
spheres or spherical patches) is also a major challenge.}''

\paragraph{Our contribution.}
We prove that any set $\T$ of $n$ disjoint triangles in $\Reals^3$ can be cut
into $O(n^{7/4}\polylog n)$ triangular fragments that admit a depth order.
Thus we overcome the first drawback of the method of Aronov, Miller and Sharir
(although admittedly our bound is not as sharp as theirs).
We also overcome the second drawback, by presenting an
algorithm to perform the cuts in $O(n^{5/2+\omega/2}\log^2 n)=O(n^{3.69})$ time.
Here $\omega<2.373$ is, as above, the exponent of the best matrix-multiplication algorithm.
As a byproduct, we improve the time to compute a complete cut set of size $O(n^{3/2}\polylog n)$ for
a collection of lines: we show that a simple trick reduces the
running time from $O(n^{4+2\omega}\log^2 n)$ to $O(n^{3+\omega}\log^2 n)$.

We also present a more refined approach that yields a bound of
$O(n^{1+\eps} + n^{1/4} K^{3/4}\polylog n)$ on the number of fragments, where
$K$ is the number of intersections between the triangles in the projection.
This result extends to $xy$-monotone surface patches bounded by a constant number
of bounded-degree algebraic arcs in general position. Thus we make
progress on all open problems posed by Aronov, Miller and Sharir.

Finally, as a minor contribution we get rid of the non-degeneracy
assumptions that Aronov and Sharir~\cite{as-atbedc-16} make when eliminating
cycles from a set of segments. Most degeneracies can be handled by a straightforward
perturbation argument, but one case---parallel segments that overlap
in the projection---requires some new ideas. Being able
to handle degeneracies for segments implies that our method for
triangles can handle degeneracies as well.

\section{Eliminating cycles among triangles}
\label{se:alg}

\paragraph{Overview of the method.}
We first prove a proposition
that gives conditions under which the existence of a depth order for a set of
triangles is implied by the existence of a depth order for the triangle edges.
The idea is then to take a complete cut set for the triangle edges---there is
such a cut set of size $O(n^{3/2}\polylog n)$ by the results of Aronov and Sharir---and
``extend'' the cuts (by taking vertical planes through the cut points)
so that the conditions of the proposition are met.
A straightforward extension would generate too many triangle fragments, however.
Therefore our cutting procedure has two phases. In the first phase we localize
the problem by partitioning space into regions such that (i) the collection
of regions admits a depth order, and (ii) each region is intersected by only
few triangles. (This localization is also the key to speeding up the algorithm
for lines.) In the second phase we then locally (inside each region)
extend the cuts from a complete cut set for the edges,
so that the conditions of the proposition are met.

\paragraph{Notation and terminology.}
Let $\T$ denote the given set of disjoint non-vertical triangles, let $\E$ denote the set
of edges of the triangles in $\T$, and let $\V$ denote the set of vertices of the triangles.
We assume the triangles in $\T$ are closed. However, at the places where a triangle is cut it becomes
open. Thus the edges of a triangle fragment that are (parts of) edges in $\E$ are part
of the fragment, while edges that are induced by cuts are not.
We denote the (vertical) projection of an object $o$ in $\Reals^3$ onto the $xy$-plane by $\proj{o}$.

\paragraph{A proposition relating depth orders for edges to depth orders for triangles.}
We define a \emph{column} to be a 3-dimensional region
$C_\Delta := \Delta\times(-\infty,+\infty)$, where $\Delta$ is an open convex polygon on the $xy$-plane.
Our cutting procedure is based on the following proposition.
\begin{proposition}\label{prop:column}
Let $\T$ be a set of disjoint triangles in $\Reals^3$, and let $\E$ be the set
of edges of the triangles in $\T$.
Let $C_\Delta$ be a column whose interior does not contain a vertex of any triangle in~$\T$,
and let $\T_\Delta := \{ T\cap C_\Delta : T\in \T \}$ and $\E_{\Delta} := \{ e\cap C : e\in \E \}$.
Then $\T_{\Delta}$ admits a depth order if $\E_{\Delta}$ admits a depth order.
\end{proposition}
\begin{proof}
For a triangle $T_i\in\T$, define $P_i := T_i \cap C$. Thus $\T_{\Delta} = \{ P_i : T_i\in\T \}$.
Assume $\E_{\Delta}$ admits a depth order and suppose for a contradiction that $\T_{\Delta}$ does not.

Consider a cycle $\C := P_0\prec P_1\prec \cdots \prec P_{k-1}\prec P_0$ in $\T_{\Delta}$.
As observed by Aronov~\etal~\cite{ams-adcat-17}
we can associate a closed curve in $\Reals^3$ to $\C$, as follows.
For each pair $P_i,P_{i+1}$ of consecutive polygons in~$\C$---here and in the rest
of the proof indices are taken modulo~$k$---let $b_i\in P_i$ and $a_{i+1}\in P_{i+1}$ be points such that
the segment~$b_i a_{i+1}$ is vertical. We refer to the
closed polygonal curve whose ordered set of vertices
is $b_1, a_2, b_2, a_3, \ldots, a_k, b_k, a_1$ as a
\emph{witness curve} for $\C$.
We call the vertical segments $b_i a_{i+1}$ the \emph{connections} of $\Gamma(\C)$,
and we call the segments $a_i b_i$ the \emph{links} of $\Gamma(\C)$. Since the
connections are vertical, we have $\proj{a_{i+1}}=\proj{b_{i}}$ and so we can write
$\proj{\Gamma(\C)}$ as $\proj{a_1},\proj{a_2},\ldots,\proj{a_k},\proj{a_1}$.
Note that $\proj{a_i}\in \proj{P_{i-1}}\cap \proj{P_i}$ for all~$i$.
In general, the points $a_i$ and $b_i$ can be chosen in many ways and so there are
many possible witness curves.
We will need a specific witness curve, as specified next.
We say that a link $a_i b_i$ is \emph{good} if $a_i$ and $b_i$
lie on the same edge of their polygon~$P_i$---this
edge is also an edge in $\E_{\Delta}$---and we say that $a_i b_i$ is \emph{bad} otherwise.
We now define the \emph{weight} of a witness curve $\Gamma$ to be the
number of bad links in $\Gamma$, and we define $\Gamma(\C)$
to be any minimum-weight witness curve for $\C$.

Now consider a minimal cycle $\C^*:= P_0 \prec P_1\prec \cdots\prec P_{k-1}\prec P_0$ in $\T_{\Delta}$.
(A cycle is minimal if any strict subset of polygons from the cycle is acyclic.)
We will argue that we
can find a cycle in $\E_{\Delta}$ consisting of edges of the polygons in $\C^*$, thus contradicting
that $\E_{\Delta}$ admits a depth order.

\begin{claiminproof}
All links $a_i b_i$ of $\Gamma(\C^*)$ are good.
\end{claiminproof}
\begin{proofinproof}
Consider any link $a_i b_i$.
Observe that $\proj{a_{i-1}}$ and $\proj{a_{i+2}}$ must both lie
outside $\proj{P_i}$, otherwise $\C^*$ is not minimal. Consider $\Delta\setminus \proj{P_i}$,
the complement of $\proj{P_i}$ inside the column base~$\Delta$. The region
$\Delta\setminus\proj{P_i}$ consists one or more connected components. (It cannot be empty
since then $P_i$ cannot be part of any cycle in $\T_{\Delta}$.)
Each connected component is separated from $\proj{P_i}$ by a single edge of~$\proj{P_i}$,
since by assumption $T_i$ does not have a vertex inside~$C$
and so $\proj{P_i}$ does not have a vertex inside~$\Delta$ either.
We now consider two cases, as illustrated in Fig.~\ref{fi:prop-fig}.
\begin{figure}
  \centering
  \includegraphics{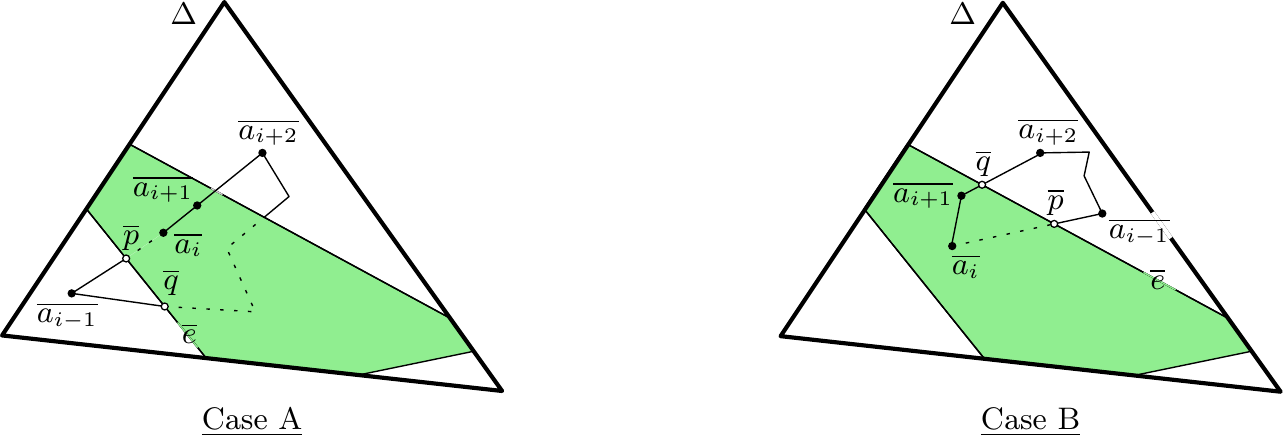}
  \caption{The two cases in the proof of Proposition~\protect\ref{prop:column}.
           Polygon $\ol{P_i}$ is shown in green.}
  \label{fi:prop-fig}
\end{figure}
\\[2mm]
\noindent \emph{Case~A: $\proj{a_{i-1}}$ and $\proj{a_{i+2}}$ lie in different components of~$\Delta\setminus\proj{P_i}$.}
    Let $\proj{p}$ be the point where $\proj{a_{i-1}a_{i}}$ enters $\proj{P_i}$.
    Let $p\in P_i$ project onto $\proj{p}$ and let $e$ be the edge of $P_i$ containing~$p$.
    (Possibly $\proj{p} = \proj{a_i}$.) Since $\proj{a_{i+2}}$ lies in a different
    connected component of $\Delta\setminus\proj{P_i}$ than $\proj{a_{i-1}}$, the projection $\proj{\Gamma(\C^*)}$
    must cross $\proj{e}$ a second time, at some point~$\proj{q}$.
    This leads to a contradicting with the minimality of~$\C^*$.
    To see this, let $q\in\Gamma(\C^*)$ be a point projecting onto~$\proj{q}$
    and let $P_j$ be such that $q\in P_j$. Then $j\not\in\{i-1,i,i+1\}$, because
    $a_{i-1}b_{i-1}$, and $a_ib_i$, and $a_{i+1}b_{i+1}$ are the only links
    of $\Gamma(\C^*)$ on $P_{i-1}$, and $P_i$, and $P_{i+1}$, respectively.
    But since $\ol{P_i}\cap\ol{P_j}\neq\emptyset$ we have $P_i\prec P_j$ or $P_j\prec P_i$,
    and so $j\not\in\{i-1,i,i+1\}$ contradicts that $\C^*$ is minimal.    Thus Case~A cannot happen.
\\[2mm]
\noindent \emph{Case~B: $\proj{a_{i-1}}$ and $\proj{a_{i+2}}$ lie in the same component of~$\Delta\setminus\proj{P_i}$.}
    In this case $a_i b_i$ must be a good link, because $a_i$ and $b_i$ must both lie on
    the edge~$e$ bordering the component of~$\Delta\setminus\proj{P_i}$ that contains $\proj{a_{i-1}}$ and $\proj{a_{i+2}}$.
    Indeed, if $\proj{a_i}$ and/or $\proj{a_{i+1}}$ would not lie on~$e$ then
    we can obtain a witness curve of lower weight for $\C^*$, namely
    if we replace $a_i$ by the point~$p$ such that
    $\proj{p} = \proj{a_{i-1} a_{i}} \cap \proj{e}$ and we replace $a_{i+1}$ by the point~$q$ such that
    $\proj{q} = \proj{a_{1} a_{i+1}} \cap \proj{e}$. (Note that if $a_{i-1}=a_{i+2}$,
    which happens when $\C^*$ consists of only three polygons, then the argument still goes through.)
\\[2mm]
\noindent Thus $a_i b_i$ is a good link, as claimed.
\end{proofinproof}
If all links $a_i,b_i$ of $\Gamma(\C^*)$ are good then $\C^*$
gives a cycle in $\E_{\Delta}$, contradicting that $\E(\sigma)$ admits a depth order.
Hence, the assumption that $\T_{\Delta}$ contains a cycle is false.
\end{proof}

\paragraph{The cutting procedure.}
A naive way to apply Proposition~\ref{prop:column} would be the following:
compute a complete cut set $X$ for the set $\E$ of triangle edges, and take
a vertical plane parallel to the $yz$-plane through each point in $\V\cup X$.
This subdivides $\Reals^3$ into columns $C_\Delta$ (where each column base $\Delta$
is an infinite strip). These columns do not contain triangle vertices
and the edge fragments inside each column are acyclic,
and so the triangle fragments we obtain are acyclic. Unfortunately this straightforward
approach generates too many fragments. Hence, we first subdivide
space such that we do not cause too much fragmentation when we take the
vertical planes through $\V\cup X$. The crucial idea is to create the subdivision
based on the projections of the triangle edges. This allows us to use an efficient
2-dimensional partitioning scheme resulting in cells that are
intersected by only few projected triangles edges. The 2-dimensional subdivision
will then be extended into $\Reals^3$, to obtain 3-dimensional regions
in which we can take vertical planes through $\V\cup X$ without
creating too many fragments. We cannot completely ignore
the triangles themselves, however, when we extend the 2-dimensional subdivision
into~$\Reals^3$---otherwise we already create too many fragments in this phase.
Thus we create a hierarchical 2-dimensional subdivision, and we use the hierarchy
to avoid cutting the input triangles into too many fragments.
Next we make these ideas precise.
\medskip

Let $L$ be a set of $n$ lines in the plane. A \emph{$(1/r)$-cutting} for $L$ is
a partition $\cut$ of the plane into triangular\footnote{The cells
may be unbounded, that is, we also allow wedges, half-planes, and the entire plane, as cells.}
cells such that the interior of any cell $\Delta\in\cut$ intersects at most $n/r$ lines from~$L$.
We say that a cutting $\cut$ \emph{$c$-refines} a cutting~$\cut'$, where $c$ is
some constant, if every cell $\Delta\in \cut$ is contained in a unique parent cell
$\Delta'\in\cut'$, and each cell in $\cut'$ contains at most $c$ cells from~$\cut$.
An \emph{efficient hierarchical $(1/r)$-cutting} for $L$~\cite{m-rsehc-93} is a sequence
$\Psi := \cut_0,\cut_1,\ldots,\cut_k$ of cuttings such that there are constants $c,\rho$ such that
the following four conditions are met:
\myenumerate{
\item[(i)] $\rho^{k-1} < r \leq \rho^k$;
\item[(ii)] $\cut_0$ is the single cell $\Reals^2$;
\item[(iii)] $\cut_i$ is a $(1/\rho^i)$-cutting for $L$ of size $O(\rho^{2i})$, for all $0\leq i\leq k$;
\item[(iv)] $\cut_i$ is a $c$-refinement of $\cut_{i-1}$, for all $1\leq i\leq k$.
}
It is known that for any set $L$ and any parameter $r$ with $1\leq r\leq n$,
an efficient hierarchical $(1/r)$-cutting exists and can be
computed in $O(nr)$ time~\cite{c-chdc-93,m-rsehc-93}.
We can view $\Psi$ as a tree in which each node $u$ at level~$i$ corresponds to a cell
$\Delta_u\in\cut_i$, and a node $v$ at level~$i$ is the child of a node~$u$ at level~$i-1$
if $\Delta_v\subseteq\Delta_u$.

Our cutting procedure now proceeds in two steps. Recall that $\T$ denotes the given set of $n$ triangles
in~$\Reals^3$, and $\E$ the set of $3n$~edges of the triangle in~$\T$.
\begin{enumerate}
\item \label{step1}
    We start by constructing an efficient hierarchical $(1/r)$-cutting for $L$, with $r=n^{3/4}$,
    where $L$ is the set of lines containing the edges in $\proj{\E}$.
    Next we cut the projection $\proj{T}$ of each triangle $T\in\T$ into pieces. This is done
    by executing the following recursive process on~$\Psi$, starting at its root.
    Suppose we reach a node~$u$ of the tree.
    If $\Delta_u\subseteq \proj{T}$ or $u$ is a leaf, then $\Delta_u\cap \proj{T}$
    is one of the pieces of $\proj{T}$.
    Otherwise, we recursively visit all children $v$ of $u$ such that $\Delta_v\cap \proj{T} \neq \emptyset$.

    After cutting each projected triangle $\proj{T}$ in this manner, we cut the original triangles
    $T\in\T$ accordingly. Let $\T_1$ denote the resulting collection of polygonal pieces.

    We extend the 2-dimensional cutting $\cut_k$ into $\Reals^3$ by erecting vertical walls through
    each of the edges in $\cut_k$. Thus we create a column $C_\Delta := \Delta\times (-\infty,\infty)$
    for each cell $\Delta\in \cut_k$. Next, we cut each column $C_\Delta$ into vertical prisms by slicing it
    with each triangle $T\in\T$ that completely cuts through $C_\Delta$ (that is, we
    slice the column with each triangle $T$
    such that $\Delta\subseteq \proj{T}$). Let~$\subdiv$ denote the resulting
    3-dimensional subdivision.
\item \label{step2}
    For each prism $\sigma$ in the subdivision~$\subdiv$, proceed as follows.
    Let $\T_1(\sigma)\subseteq \T_1$ be the set of pieces that have an edge intersecting
    the interior of~$\sigma$, and let $\E(\sigma) := \{ e\cap \sigma : e\in\E\}$.
    Note that $\E(\sigma)$ is the set of edges of the pieces in $\T_1(\sigma)$,
    where we only take the edges in the interior of $\sigma$.
    Let $X(\sigma)$ be a complete cut set for $\E(\sigma)$,
    and let $\V(\sigma)\subseteq \V$ be the set of triangle vertices in the interior of~$\sigma$.
    For each point $q\in X(\sigma)\cup \V(\sigma)$, take a plane $h(q)$ containing $q$ and
    parallel to the $yz$-plane, and let $H(\sigma)$ be the resulting set of planes.
    Cut every piece $P\in \T_1(\sigma)$ into fragments using the planes in $H(\sigma)$.
\end{enumerate}
We denote the set of fragments generated in Step~\ref{step2} inside a prism~$\sigma$
by $\T_2(\sigma)$, and we denote the set of pieces in $\T_1$ that do not have an edge
crossing the interior of any prism~$\sigma\in\subdiv$ by $\T_1^*$.
(Note that $\T_1^{*}$ contains all pieces generated at internal nodes of $\Psi$.)
Then $\T_2 := \T_1^* \cup \bigcup_{\sigma\in\subdiv} \T_2(\sigma)$ is our final set of
fragments.
\begin{lemma}\label{le:number-of-fragments}
The set $\T_2$ of triangle fragments resulting from the procedure above is acyclic,
and the size of $\T_2$ is $O(n^{7/4}+ |X| \cdot n^{1/4})$, where
$X := \bigcup_{\sigma\in\subdiv} X(\sigma)$.
\end{lemma}
\begin{proof}
To prove that $\T_2$ is acyclic, define $\subdiv^*$ to be the set of (open) prisms
in $\subdiv$, and consider the set $\subdiv^* \cup \T_1^{*}$.

\begin{claiminproof}
The set $\subdiv^* \cup \T_1^{*}$ admits a depth order.
\end{claiminproof}
\begin{proofinproof}
By construction, for any object~$o_i\in \subdiv^* \cup \T_1^{*}$ there is
a node $u\in\Psi$ such that $\proj{o_i} = \Delta_u$. Hence,
for any two objects $o_1,o_2\in \subdiv^* \cup \T_1^{*}$ we have
\begin{equation}
\proj{o_1}\subseteq\proj{o_2}, \hspace{5mm} \mbox{or} \hspace{5mm}
\proj{o_2}\subseteq\proj{o_1}, \hspace{5mm} \mbox{or} \hspace{5mm}
\proj{o_1}\cap\proj{o_2}=\emptyset. \label{eq1}
\end{equation}
This implies that $\subdiv^* \cup \T_1^{*}$ is acyclic.
Indeed, suppose for a contradiction
that $\subdiv^* \cup \T_1^{*}$ does not admit a depth order. Consider a minimal
cycle $\C^*:= o_1 \prec o_2\prec \cdots\prec o_{k}\prec o_1$. Obviously $k\geq 3$.
But then (\ref{eq1}) implies that we can remove $o_1$ or $o_2$ and still
have a cycle, contradicting the minimality of~$\C^*$.
\end{proofinproof}
The claim above implies that $\T_2$ is acyclic if each of the sets $\T_2(\sigma)$ is acyclic.
To see that $\T_2(\sigma)$ is acyclic, note that the planes in $H(\sigma)$
partition~$\sigma$ into subcells that do not contain a point from $X(\sigma)$
in their interior. Hence the set of edges of the fragments in such a subcell
is acyclic---if this were not the case, then there would be a cycle left
in $\E(\sigma)$, contradicting that $X(\sigma)$ is a complete cut set for $\E(\sigma)$.
Moreover, a subcell does not contain any point from $\V(\sigma)$ in its interior,
and so it does not contain a vertex of any fragment in its interior.
We can therefore use Proposition~\ref{prop:column} to conclude that within each
subcell, the fragments are acyclic; the fact that the subcell is strictly
speaking not a column---it may be bounded from above and/or
below by a piece in $\T_1^*$---clearly does not invalidate the conclusion.
Since the fragments in each subcell of $\sigma$ are acyclic and the subcells
are separated by vertical planes, $\T_2(\sigma)$ must be acyclic.
\medskip

It remains to prove that $|\T_2|=O(n^{7/4}+ |X|\cdot n^{1/4})$.
We start by bounding~$|\T_1|$. To this end, consider a triangle $T\in \T$
and let $P\in \T_1$ be a piece generated for $T$ in Step~\ref{step1}. Let $v$ be the
node in $\Psi$ where $P$ was created. Then the cell $\Delta_u$ of the parent $u$ of $v$ is intersected by
an edge of $\proj{T}$. Since each node in $\Psi$ has $O(1)$ children and each
cell $\Delta\in\Xi_i$ intersects at most $n/\rho^i$ projected triangle edges, this means that
\[
|\T_1|  =  O\left( \sum_{i=0}^{k-1} \sum_{\Delta\in\cut_i} n/\rho^i  \right)
        =  O\left( \sum_{i=0}^{k-1} \rho^{2i} \cdot (n/\rho^i)  \right)
        =  O(n \rho^{k})
        =  O(n r)
        =  O(n^{7/4}).
\]
The number of additional fragments created in Step~\ref{step2} can be bounded by
observing that each prism $\sigma$ in the subdivision~$\subdiv$ intersects at most
$n/r = O(n^{1/4})$ triangle edges, and so $|T_1(\sigma)|=O(n^{1/4})$. If we now sum
the number of additional fragments over all prisms $\sigma$ in the subdivision~$\subdiv$ we obtain
\[
\begin{array}{lll}
\mbox{number of additional fragments in Step~\ref{step2}} & \leq & \sum_{\sigma\in\subdiv} |H(\sigma)| \cdot |\T_1(\sigma)| \\[2mm]
    & \leq & O(n^{1/4}) \cdot \big( \sum_{\sigma\in\subdiv} |X(\sigma)| + \sum_{\sigma\in\subdiv} |\V(\sigma)| \big) \\[2mm]
    &  =  & O(n^{1/4} (|X|+n) ).
\end{array}
\]
\end{proof}
Lemma~\ref{le:number-of-fragments} leads to the following result.
\begin{corollary}
Suppose that any set of $n$ lines has a complete cut set of size $\gamma(n)$.
Then any set $\T$ of $n$ disjoint triangles in $\Reals^3$
can be cut into $O(n^{7/4} + \gamma(3n)\cdot n^{1/4})$ triangular fragments such that the
resulting set of fragments admits a depth order.
\end{corollary}
\begin{proof}
Define $\opt$ to be the minimum size of a complete cut set for $\E$ and,
for a prism~$\sigma\in\subdiv$, define $\opt_\sigma$ to be the minimum size of a complete cut for
$\E(\sigma)$. Then $\sum_{\sigma\in\subdiv} \opt_\sigma \leq \opt$. Indeed, if $\Xopt$
denotes a minimum-size complete cut set for~$\E$, then $\Xopt\cap \sigma$ must eliminate
all cycles from $\E(\sigma)$. Since $\opt \leq \gamma(3n)$, the bound on the number of fragments
generated by our cutting procedure is as claimed.

The procedure above cuts the triangles in $\T$ into constant-complexity polygonal
fragments, which we can obviously cut further into triangular fragments without increasing the
number of fragments asymptotically.
\end{proof}
The results of Aronov and Sharir~\cite{as-atbedc-16} thus imply that any set of
$n$ triangles can be cut into $O(n^{7/4}\polylog n)$ fragments such that the resulting
set of fragments is acyclic. (Aronov and Sharir assume general position, but in the
appendix we show this is not necessary.)

\begin{figure}
  \centering
  \includegraphics{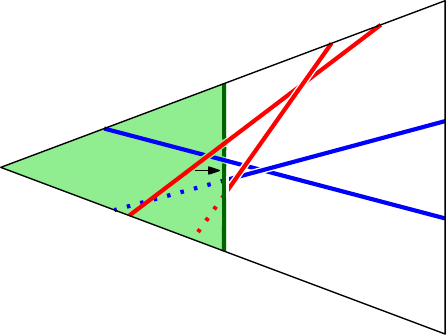}
  \caption{An example showing that one sometimes needs more cuts to eliminate all cycles
            from $\T_1(\sigma')$ than from $\E(\sigma')$. The set $\T_1(\sigma')$ consists
            of the green triangle, and two red and two blue segments.
            (The red and blue segments can easily be replaced by very thin triangles.)
            The set $\E(\sigma')$ consists of the dark green edge of the green triangle,
            and the red and blue segments.
            Note that the configuration shown in the figure is realizable.}
  \label{fi:edges-vs-triangles}
\end{figure}
\begin{remark}
We use a factor $O(n^{1/4})$ more cuts than Aronov and Sharir need for the case of segments.
Observe that we already generate up to
$\Theta(n^{7/4})$ fragments in Step~\ref{step1}, since we take $r=n^{3/4}$.
To reduce the total number of fragments to $O(n^{3/2}\polylog n)$ using our approach,
we would need to set $r:=\sqrt{n}$ in Step~\ref{step1}. In Step~\ref{step2}
we could then only use the set $\V(\sigma)$ to generate the vertical planes in~$H(\sigma)$.
This would lead to vertical prisms that do not have any vertex in their interior,
while only using $O(n^{3/2})$ fragments so far. However, each such
prism~$\sigma'\subseteq \sigma$ can contain
up to $\Theta(\sqrt{n})$ triangle fragments. Hence, we cannot afford to compute
a cut set $X(\sigma')$ for $\E(\sigma')$ and cut each triangle fragment in $\sigma'$
with a vertical plane containing each $q\in X(\sigma')$. One may hope that if we can
eliminate all cycles from $\E(\sigma')$ using $|X(\sigma')|$ cuts, then we can also eliminate
all cycles from $\T_1(\sigma')$ using $|X(\sigma')|$ cuts. Unfortunately this is not the
case, as shown in Fig.~\ref{fi:edges-vs-triangles}.

In the example, there are two cycles in $\T_1(\sigma')$: the green triangle together with the blue segments
and the green triangle with the red segments. The set $\E(\sigma')$
also contains two cycles. The cycles from $\E(\sigma')$ can be eliminated
by cutting the green edge at the point indicated by the arrow. However, a single
cut of the green triangle cannot eliminate both cycles from~$\T_1(\sigma')$.
Indeed, to eliminate the blue-green cycle the cut should separate
(in the projection) the parts of the blue edges projecting onto the green triangle,
while to eliminate the red-green cycle the cut should separate
the parts of the red edges projecting onto the green triangle---but a single
cut cannot do both.
The example can be generalized to sets $\T_1(\sigma')$ of arbitrary
size, so that all cuts in $\E(\sigma')$ can be eliminated by a single cut,
while eliminating cycles from $\T_1(\sigma')$ requires $\Omega(|\T_1(\sigma')|)$ cuts.
Thus a more global reasoning is needed to improve our bound.
\end{remark}

\section{Efficient algorithms to compute complete cut sets}

\paragraph{The algorithm for triangles.}
The hierarchical cutting $\Psi$ can be computed in $O(nr)=O(n^{7/4})$ time~\cite{c-chdc-93,m-rsehc-93},
and it is easy to see that we can compute the set $\T_1$ within the same
time bound. Constructing the 3-dimensional subdivision~$\subdiv$ can trivially be done
in $O(n^{5/2})$ time, by checking for each of the $O(n^{3/2})$ columns and each triangle
$T\in \T$ if $T$ slices the column. Next we need to find the sets $\T_1(\sigma)$
for each prism~$\sigma$ in~$\subdiv$. The computation of the hierarchical cutting also
tells us for each cell $\Delta\in\cut_k$ which projected triangle edges intersect~$\Delta$.
It remains to check, for each triangle $T$ corresponding to such an edge, which
of the $O(n)$ prisms of the column~$C_{\Delta}$ is intersected by~$T$.
Thus, we spend $O(n)$ time for each of the $O(n^{7/4})$ triangle pieces in $\T_1$,
so the total time to compute the sets $\T_1(\sigma)$ is $O(n^{11/4})$.
\medskip

Next we need to compute the cut sets $X(\sigma)$. To this end we use
the algorithm by Aronov~\etal~\cite{abgm-ccrs-08}, which computes a
complete cut set of size $O(\opt_\sigma \cdot \log \opt_\sigma \cdot\log\log \opt_\sigma)$,
where $\opt_\sigma$ is the minimum size of a complete cut set for $\E(\sigma)$.
Thus $|X|$, the total size of all cut sets $X(\sigma)$ we compute, is bounded by
\[
O\left( \sum_{\sigma\in\subdiv} \opt_{\sigma} \cdot \log \opt_{\sigma} \cdot\log\log \opt_{\sigma} \right)
=
O(\opt \cdot \log \opt \cdot\log\log \opt)
=
O(n^{3/2}\polylog n).
\]
Now define $n_\sigma := |T_1(\sigma)|$.
Since the algorithm of Aronov~\etal runs in time $O(m^{4+2\omega}\log^2 m)$
for $m$ segments, the total running time is
\[
O\left( \sum_{\sigma\in\subdiv} n_{\sigma}^{4+2\omega}\log^2 n_\sigma \right).
\]
Since $n_{\sigma}=O(n^{1/4})$ for all $\sigma$
and $\sum_{\sigma\in\subdiv} n_{\sigma} = O(n^{7/4})$, the total time to compute the
sets $X(\sigma)$ is
\[
O\left( \sum_{\sigma\in\subdiv} n_{\sigma}^{4+2\omega}\log^2 n_\sigma \right)
=
O\left( n^{3/2} \cdot (n^{1/4})^{4+2\omega}\log^2 n \right)
=
O(n^{5/2+\omega/2}\log^2 n).
\]
Finally, for each prism~$\sigma$ we cut all triangles in $\T_1(\sigma)$
by the planes in $H(\sigma)$ in a brute-force manner, in total time~$O(n^{7/4}\polylog n)$.

The following theorem summarizes our main result.
\begin{theorem}\label{thm:triangles}
Any set $\T$ of $n$ disjoint non-vertical triangles in $\Reals^3$
can be cut into $O(n^{7/4}\polylog n)$ triangular fragments such that the
resulting set of fragments admits a depth order. The time needed to compute the
cuts is $O(n^{5/2+\omega/2}\log^2 n)$, where $\omega<2.373$ is the exponent
in the running time of the best matrix-multiplication algorithm.
\end{theorem}

\paragraph{A fast algorithm for lines.}
The running time in Theorem~\ref{thm:triangles} is better than the running time
obtained by Aronov and Sharir~\cite{as-atbedc-16} to compute a complete cut set
for a set of lines in~$\Reals^3$. The reason is that we apply
the algorithm of Aronov~\etal\cite{abgm-ccrs-08} locally, on a set of segments
whose size is significantly smaller than~$n$. We can use the same idea to speed
up the algorithm to compute a complete cut
set of size $O(n^{3/2}\polylog n)$ for a set $L$ of $n$ lines in $\Reals^3$.
To this end we project $L$ onto the $xy$-plane, and compute a $(1/r)$-cutting $\cut$
for $\proj{L}$ of size~$O(r^2)$, with $r:=\sqrt{n}$.
We then cut each
line $\ell\in L$ at the points where its projection~$\proj{\ell}$ is cut by the cutting
(that is, where $\proj{\ell}$ crosses the boundary of a cell~$\Delta$ in~$\cut$).
Up to this point we make only $O(nr)=O(n^{3/2})$ cuts, which does not
affect the worst-case asymptotic bound on the number of cuts.

Each cell~$\Delta$ of the cutting defines a vertical column $C_{\Delta}$.
Within each column, we apply the algorithm of Aronov~\etal~\cite{abgm-ccrs-08}
to compute a cut set of size
$O(\opt_{{\Delta}} \cdot \log \opt_{{\Delta}} \cdot\log\log \opt_{{\Delta}})$,
where $\opt_{{\Delta}}$ is the size of an optimal cut set inside the column.
In total this gives
$O(\opt \cdot \log \opt \cdot\log\log \opt) = O(n^{3/2}\polylog n)$ cuts
in time $O(n \cdot (n^{1/2})^{4+2\omega}) = O(n^{3+\omega})$.

This leads to the following result.
\begin{theorem}\label{thm:lines}
For any set $L$ of $n$ disjoint lines in $\Reals^3$, we can compute
in $O(n^{3+\omega})$ time a set of $O(n^{3/2}\polylog n)$ cut points on the lines such that the
resulting set of fragments admits a depth order, where $\omega<2.373$ is the exponent
in the running time of the best matrix-multiplication algorithm.
\end{theorem}

\section{A more refined bound and an extension to surface patches}
Let $\T$ be a set of disjoint surface patches in~$\Reals^3$. We assume each surface
patch is $xy$-monotone, that is, each vertical line intersects a patch in a single
point or not at all, and we assume each surface patch is bounded by a constant number
of bounded-degree algebraic arcs. We refer to the arcs bounding a surface patch as
the \emph{edges} of the surface patch. We assume the edges are in general position
as defined by Aronov and Sharir~\cite{as-atbedc-16}, except that edges of the same
patch may share endpoints.
We will show how to cut the patches from $\T$ into
fragments such that the resulting fragments admit a depth order. The total
number of fragments will depend on~$K$, the number of intersections between the
projections of the edges: for any fixed $\eps>0$, we can tune our procedure
so that it generates $O(n^{1+\eps} + n^{1/4} K^{3/4} \polylog n)$ fragments.
Trivially this implies that the same
intersection-sensitive bound holds for triangles.
\medskip

The extension of our procedure to obtain an intersection-sensitive
bound for surface patches is fairly straightforward.
First we observe that the analog of Proposition~\ref{prop:column} still holds, where
the base of the column can now have curved edges. In fact, the proof holds verbatim,
if we allow the links of the witness cycles $\Gamma(\C)$ that connect points
$a_i$ and $b_i$ on the same surface patch to be curved.
Now, instead of using efficient hierarchical cuttings~\cite{c-chdc-93,m-rsehc-93}
we recursively generate a sequence of cuttings using the intersection-sensitive
cuttings of De~Berg and Schwarzkopf~\cite{bs-ca-95}. (This is somewhat similar
to the way in which Aronov and Sharir~\cite{as-atbedc-16} obtain an intersection-sensitive
bound on the number of cuts needed to eliminate all cycles for a set of line segments in $\Reals^3$.)
Below we give the details.

Let $\E$ denote the set of $O(n)$ edges of the surface patches in~$\T$.
A $(1/r)$-cutting for $\proj{\E}$ is a subdivision of $\Reals^2$ into
trapezoidal cells, such that the interior of each cell intersects at most $n/r$
edges from~$\proj{E}$. Here a trapezoidal cell
is a cell bounded by at most two segments that are parallel to the $y$-axis and at most two pieces of
edges in~$\proj{\E}$ (at most one bounding it from above and at most one bounding it from below).
Set $r := \min(n^{5/4}/K^{1/4},n)$. Let $\rho$ be a sufficiently large constant, and let $k$ be
such that $\rho^{k-1}<r\leq \rho^k$; the exact value of $\rho$ depends on the
desired value of $\eps$ in the final bound. We recursively construct a hierarchy
$\Psi := \cut_0,\cut_1,\ldots,\cut_k$ of cuttings such that $\cut_i$ is a
$(1/\rho^i)$-cutting for~$\proj{\E}$, as follows. The initial cutting~$\cut_0$ is
the entire plane~$\Reals^2$. To construct $\cut_i$ we take each cell $\Delta$
of $\cut_{i-1}$ and we construct a $(1/\rho)$-cutting for the set
$\proj{\E}_\Delta := \{ \proj{e}\cap \Delta : \proj{e}\in\proj{\E}\}$.
De~Berg and Schwarzkopf~\cite{bs-ca-95} have shown that there is such a cutting
consisting of $O(\rho + K_{\Delta}\rho^2/n_\Delta^2)$ cells, where $n_\Delta := |\proj{\E}_\Delta|$
and $K_{\Delta}$ is the number of intersections inside~$\Delta$. One easily
shows by induction on~$i$ that for each cell $\Delta$ in $\cut_{i-1}$ we have\footnote{Strictly
speaking this is not true, as $n$ denotes the number of patches and not the
number of edges. To avoid cluttering the notation we allow ourselves this
slight abuse of notation.}
$n_\Delta \leq n/\rho^{i-1}$. Hence, by combining the cuttings $\cut_\Delta$
over all $\Delta\in\cut_{i-1}$ we obtain a $(1/\rho^i)$-cutting~$\cut_i$.

Let $|\cut_i|$ be the number of cells in $\cut_i$. Then $|\cut_0|=1$ and, for a suitable
constant~$D$ (which depends on the degree of the edges and follows from
the construction of De~Berg and Schwarzkopf~\cite{bs-ca-95}),
we have
\[
\begin{array}{llll}
|\cut_i| & \leq & \sum_{\Delta\in\cut_{i-1}} D (\rho + K_{\Delta}\rho^2/n_\Delta^2) & \\[2mm]
    & \leq & D\rho \cdot |\cut_{i-1}| + D K\rho^{2i}/n^2 & \mbox{(since $\sum_{\Delta\in\cut_{i-1}} K_\Delta\leq K$ and $n_\Delta \leq n/\rho^{i-1}$)} \\[2mm]
        & \leq & D^i \rho^i + \frac{DK}{n^2} \sum_{j=0}^i D^{j} \rho^{2i-j}  & \\[2mm]
    & \leq & D^i \rho^i + \frac{DK}{n^2} \cdot 2\rho^{2i} & \mbox{(assuming $\rho>2D$).}
\end{array}
\]
Now we can proceed exactly as before. Thus we first traverse the hierarchy~$\Psi$
with each patch~$T\in\T$, associating $T$ to nodes $u$ such that $\Delta_u\subseteq\proj{T}$
and $\Delta_{\pa(u)}\not\subseteq\proj{T}$, and to
the leaves that we reach. This partitions $T$ into a number of fragments.
The resulting set~$\T_1$ of fragments generated over all triangles $T\in\T$ has total size

\begin{eqnarray}
|\T_1|  & = & O\left( \textstyle{\sum}_{i=0}^{k-1} \textstyle{\sum}_{\Delta\in\cut_i} n/\rho^i  \right) \label{eq:T1} \\
        & = & O\left( \textstyle{\sum}_{i=0}^{k-1} \left( D^i \rho^i + \frac{DK}{n^2} \rho^{2i} \right) \cdot (n/\rho^i)  \right) \nonumber \\
        & = & O\left( n \textstyle{\sum}_{i=0}^{k-1}D^i + \frac{DK}{n} \textstyle{\sum}_{i=0}^{k-1} \rho^{i} \right) \nonumber \\
        & = & O\left( n D^k + DKr/n \right) \nonumber
\end{eqnarray}

If we now set $\rho := D^{1/\eps}$ then $D^k =\rho^{k\eps} = O(r^{\eps})=O(n^{\eps})$, and so
$|\T_1|=O(n^{1+\eps} + Kr/n)$. We then decompose $\Reals^3$ into a subdivision~$\subdiv$ 
consisting of prisms~$\sigma$
that each intersect at most $n/r$ surface patches, take a minimum-size complete cut set~$X(\sigma)$
for the edges inside each prism~$\sigma$, and generate a set $H(S)$ of cutting planes
through the points in $X(\sigma)\cup \V(\sigma)$. (Here $\V(\sigma)$ is, as before,
the set of vertices of the surface patches in the interior of~$\sigma$.)
Since Aronov and Sharir~\cite{as-atbedc-16} proved that any set of $n$ bounded-degree
algebraic arcs in general position\footnote{We assumed the edges are in general
position, but edges of the same patch may share endpoints. However, there
are no endpoints in the interior of~$\sigma$, and so the only degeneracy that can
happen is if two edges of the same patch share an endpoint that lies on the
boundary of~$\sigma$. In this case we can slightly shorten and perturb the
edges to remove this degeneracy as well.}
admits a complete cut set of size $O(n + (nK)^{1/2}\polylog n)$---the
constant of proportionality and the exponent of the polylogarithmic factor depend
on the degree of the arcs---we have
\[
\sum_{\sigma\in\subdiv} |X(\sigma)| = O\left( n + (nK)^{1/2}\polylog n \right)
\]
and so the number of additional fragments created in Step~\ref{step2} is bounded by
\[
O(n/r) \cdot \left( \sum_{\sigma\in\subdiv} (|\V(\sigma)|+ |X(\sigma)|) \right)
=
O\left( n^2/r + (n^{3/2}K^{1/2}/r)\polylog n \right).
\]
By picking $r := \min(n^{5/4}/K^{1/4},n)$ our final bound on the number of fragments becomes
$O(n^{1+\eps} + n^{1/4} K^{3/4}\polylog n)$.

\begin{theorem}\label{thm:patches}
Let $\T$ be a set of $n$ disjoint $xy$-monotone surface patches in $\Reals^3$, each
bounded by a constant number of constant-degree algebraic arcs in general position.
Then for any fixed $\eps>0$ we can cut $\T$ into $O(n^{1+\eps} + n^{1/4} K^{3/4}\polylog n)$ fragments
that admit a depth order, where $K$ is the number
of intersections between the projections of the edges of the surface patches in~$\T$.
The constant of proportionality and the exponent of the polylogarithmic factor depend
on the degree of the edges. The expected time needed to compute the cuts is 
$O\left( n^{1+\eps} + K^{(3+\omega)/2+\eps} / n^{(1+\omega)/2}  \right)$, 
where $\omega<2.373$ is the exponent
in the running time of the best matrix-multiplication algorithm.
\end{theorem}
\begin{proof}
The bound on the number of fragments follows from the discussion above. To prove the time
bound we first note that an intersection-sensitive cutting of size 
$O(\rho + K_{\Delta}\rho^2/n_\Delta^2)$ can be computed in expected time
$O(n_\Delta \log \rho + K_\Delta \rho/n_\Delta)$~\cite{bs-ca-95}.
Hence, constructing the hierarchy takes expected time
\[
\begin{array}{lll}
O\left( \sum_{i=0}^{k-1} \sum_{\Delta\in\cut_i} \left( \frac{n}{\rho^i}\log \rho  + K_{\Delta} \frac{\rho}{ (n/\rho^i)} \right) \right)
        & = & 
O\left( \sum_{i=0}^{k-1} \sum_{\Delta\in\cut_i} \frac{n}{\rho^i}\log \rho \right) + 
O\left( \sum_{i=0}^{k-1} K \frac{\rho^{i-1}}{n} \right).
\end{array}
\]
The first term is the same as in Equation~(\ref{eq:T1})
except for the extra $\log\rho$-factor (which is a constant), so this term is still bounded by 
$O(n^{1+\eps} + Kr/n)$. Since $\rho^{k-1}\leq r$, the second term is bounded by $O(Kr/n)$,
which is dominated by the first term. Thus the total expected time to compute the set $\T_1$ is
$O(n^{1+\eps} + Kr/n)$.

In the second stage of the algorithm we use the algorithm of Aronov~\etal~\cite{abgm-ccrs-08}
on the set $\E(\sigma)$ of edge fragments inside each prism~$\sigma\in \subdiv$. Aronov~\etal
only explicitly state their result for line segments, but it is easily checked that
it works for curves as well; the fact that, for example, there can already be cyclic
overlap between a pair of curves has no influence on the algorithm's approximation
factor or running time. (The crucial property still holds that cut
points can be ordered linearly along a curve, and this is sufficient for the algorithm to work.)
Thus the time needed to compute all cut sets $X(\sigma)$ is
\[
O\left( \sum_{\sigma\in\subdiv} n_{\sigma}^{4+2\omega}\log^2 n_\sigma \right),
\]
where $n_\sigma$ is the number of edges inside~$\sigma$.
Since $n_\sigma = O(n/r)$ and $\sum_{\sigma\in\subdiv} n_{\sigma} = O(|\T_1|) = O(n^{1+\eps} + Kr/n)$,
computing the cut sets takes 
\[
O\left( \frac{n^{4+2\omega+\eps}}{r^{3+2\omega}} + K \left( \frac{n}{r} \right)^{2+2\omega+\eps} \right)
\] 
time (for a slightly larger $\eps$ than before).
Because we picked $r := \min(n^{5/4}/K^{1/4},n)$, the time to compute the cut sets is
\[
O\left( n^{1+\eps} + n^{\frac{1}{4}-\omega/2+\eps}\cdot K^{\frac{3}{4}+\omega/2} + \frac{K^{\frac{3}{2}+\omega/2}+\eps}{n^{\frac{1}{2}+\omega/2}}  \right)
\ \ = \ \
O\left( n^{1+\eps} + K^{\frac{3}{2}+\omega/2+\eps} / n^{\frac{1}{2}+\omega/2}  \right),
\] 
which dominates the time for the first stage.
Finally, cutting the patches inside each prism~$\sigma$ in a brute-force manner
takes time linear in the maximum number of fragments we generate, namely
$O(n^{1+\eps} + n^{1/4} K^{3/4}\polylog n)$. Thus the total expected time
is 
\[
O\left( n^{1+\eps} + K^{\frac{3}{2}+\omega/2+\eps} / n^{\frac{1}{2}+\omega/2}  \right).
\]
Observe that for $K=n^2$ the bound we get essentially the same bound as in
Theorem~\ref{thm:triangles}.
\end{proof}

\section{Dealing with degeneracies}
We first show how to deal with degeneracies when eliminating cycles from a set of segments
(or lines) in~$\Reals^3$, and then we argue that the method for triangles presented in
the main text does not need any non-degeneracy assumptions either. (We do not deal with
removing the non-degeneracy assumptions for the case of surface patches.)

\paragraph{Degeneracies among segments.}
Let $S=\{s_1,\ldots,s_n\}$ be a set of disjoint segments in~$\Reals^3$.
(Even though we allow degeneracies, we do not allow the segments in $S$ to intersect or touch,
since then the problem is not well-defined. If the segments
are defined to be relatively open, then we can also allow an
endpoint of one segment to coincide with an endpoint of, or lie in the interior of, another segment.)
We can assume without loss of generality that $S$ does not contain vertical segments,
since eliminating all cycles from the non-vertical segments in $S$ also eliminates all
cycles when we include the vertical segments.
Aronov and Sharir~\cite{as-atbedc-16} make the following non-degeneracy assumptions:
\myenumerate{
\item[(i)] no endpoint of one segment projects onto any other segment;
\item[(ii)] no three segments are concurrent (that is, pass through a common point) in the projection;
\item[(iii)] no two segments in $S$ are parallel.
}
The main difficulty arises from type~(iii) degeneracies where
parallel segments overlap in the projection. The problem is that a small
perturbation will reduce the intersection in the projection to a single point,
and cutting one of the segments at the intersection is effective for the perturbed
segments but not necessarily for the original segments. Next we describe how we handle this
and how to deal with the other degeneracies as well.
\medskip

First we slightly extend each segment in $S$---segments that are relatively open
 would be slightly shortened---to get rid of degeneracies of type~(i),
and we slightly translate each segment to make sure no two segments
intersect in more than a single point in the projection. (The translations are not necessary, but they
simplify the following description and bring out more clearly how
the $\prec$-relations between parallel segments are treated.)
Next, we slightly perturb each segment such that all degeneracies disappear and
any two non-parallel segments whose projections intersect
before the perturbation still do so after the perturbation. This gets rid of
degeneracies of types~(ii) and~(iii). Let $s'_i$ denote the segment $s_i$
after the perturbation, and define $S' := \{s'_1,\ldots,s'_n\}$. The set $S'$
has the following properties:
\myitemize{
\item for any two non-parallel segments $s_i,s_j\in S$ we have $s_i\prec s_j$
    if and only if $s'_i\prec s'_j$;
\item the order of intersections along segments in the projection is preserved in the following sense:
    if $\proj{s'_i}\cap\proj{s'_j}$ lies before $\proj{s'_i}\cap\proj{s'_k}$ along $\proj{s'_i}$
    as seen from a given endpoint of $\proj{s'_i}$, then $\proj{s_i}\cap\proj{s_j}$ does
    not lie behind $\proj{s_i}\cap\proj{s_k}$ along $\proj{s_i}$ as seen from the
    corresponding endpoint of~$\proj{s_i}$;
\item if $s_i$ and $s_j$ are parallel then $\proj{s'_i}$ and $\proj{s'_j}$ do not intersect.
}
We will show how to obtain
a complete cut set for $S$ from a complete cut set $X'$ for $S'$.
The cut set for $S$ will consist of a cut set $X$ that is derived from $X'$ plus
a set $Y$ of $O(n\log n)$ additional cuts, as explained next.

\begin{itemize}
\item Let $q'\in X'$ be a cut point on a segment $s'_i\in S'$.
    Let $s'_j\in S'$ be the segment such that $\proj{s'_i}\cap\proj{s'_j}$
    is the intersection point on $\proj{s'_i}$ closest to $\proj{q'}$,
    with ties broken arbitrarily. (We can assume that $s'_j$ exists, since
    if $\proj{s'_i}$ does not intersect any projected segment then
    the cut point~$q$ is useless and can be ignored.) Now we put into $X$ the
    point $q\in s_i$ such that $\proj{q}=\proj{s_i}\cap\proj{s_j}$.
    (It can happen that several cut points along $s'_i$ generate the same
    cut point along~$s_i$. Obviously we need to insert only one of them into~$X$.)
    The crucial property of the cut point $q\in X$ generated for $q'\in X'$ is the following:
    \myitemize{
    \item if $\proj{q'}$ coincides with a certain intersection
        along $\proj{s'_i}$ then $\proj{q}$ coincides with the corresponding intersection
        along $\proj{s_i}$;
    \item if $\proj{q'}$ separates two intersections along
        $\proj{s'_i}$ then $\proj{q}$ separates the corresponding intersections along~$\proj{s_i}$
        or it coincides with at least one of them.
    }
    By treating all cut points in $X'$ in this manner, we obtain the set~$X$.
\item The set~$Y$ deals with parallel segments in $S$ whose projections overlap.
    It is defined as follows. Let $S(X)$ be the set of fragments resulting from
    cutting the segments in $S$ at the cut points in~$X$. Partition $S(X)$ into
    subsets~$S_{\ell}(X)$ such that $S_\ell(X)$ contains all fragments from $S(X)$
    projecting onto the same line~$\ell$. Consider such a subset $S_{\ell}(X)$
    and assume without loss of generality that $\ell$ is the $x$-axis.
    Construct a segment tree~\cite{bcko-cgaa-08} for the projections of the fragments in~$S_\ell(X)$.
    Each projected fragment $\proj{f}$ is stored at $O(\log |S_\ell(X)|)=O(\log n)$
    nodes of the segment tree,
    which induces a subdivision of $\proj{f}$ into $O(\log n)$ intervals.
    We put into $Y$ the $O(\log n)$ points on $f$ whose projections define these intervals.
    The crucial property of segment trees that we will need is the following:
    \myitemize{
    \item Let $I_v$ denote the interval corresponding to a node~$v$. Then for any two
        nodes~$v,w$ we either have $I_v\subseteq I_w$ (when $v$ is a descendent of $w$),
        or we have $I_v\supseteq I_w$ (when $v$ is an ancestor of $w$), or otherwise
        the interiors of $I_v$ and $I_w$ are disjoint. Hence, a similar property holds
        for the projections of the sub-fragments resulting from cutting the fragments
        in $S_\ell(X)$ as explained above.
    }

    Doing this for all fragments $s_i\in S_\ell(X)$ and for all subsets $S_\ell(X)$ gives us
    the extra cut set~$Y$.
\end{itemize}
\begin{lemma}\label{le:remove-degen}
The set $X\cup Y$ is a complete cut set for $S$.
\end{lemma}
\begin{proof}
Let $F$ denote the set of fragments resulting from cutting the segments in~$S$
at the points in~$X\cup Y$, and suppose for a contradiction that $F$ still contains a cycle.
Let $\C:=f_0\prec f_1\prec\cdots\prec f_{k-1}\prec f_0$ be a minimal cycle in $F$,
and let $s_i\in S$ be the segment containing~$f_i$.

As explained above, the cut points in $Y$ guarantee that for any two
parallel fragments in~$F$ whose projections overlap, one is contained in the other in the projection.
This implies that two consecutive fragments $f_i,f_{i+1}$ in $\C$ cannot be parallel: if they were,
then $\proj{f_i}\subseteq \proj{f_{i+1}}$ (or vice versa) which contradicts that
$\C$ is minimal. Hence, any two consecutive fragments are non-parallel. Now consider
the witness curve $\Gamma(\C)$ for $\C$. Since consecutive fragments in $\C$ are
non-parallel, $\Gamma(\C)$ is unique. Let $\Gamma'$ be the corresponding curve
for $S'$, that is, $\Gamma'$ visits the segments $s'_0,s'_1,\ldots,s'_{k-1},s'_0$
from $S'$ in the given order---recall that $f_i\subseteq s_i$ and that $s'_i$ is
the perturbed segment $s_i$---and it steps from $s'_i$ to $s'_{i+1}$ using
vertical connections. Since $X'$ is a complete cut set for $S'$, there must be a
link of $\Gamma'$, say on segment $s'_i$, that contains a cut point~$q'\in X'$.
In other words, $\proj{q'}$ separates $\proj{s'_{i-1}}\cap \proj{s'_{i}}$ from
$\proj{s'_{i}}\cap \proj{s'_{i+1}}$, or it coincides with one of these points.
But then the cut point~$q\in X$ corresponding to $q'$ must separate
$\proj{s_{i-1}}\cap \proj{s_{i}}$ from $\proj{s_{i}}\cap \proj{s_{i+1}}$
or coincide with one of these points, thus cutting the witness curve $\Gamma(\C)$---a contradiction.
\end{proof}

\begin{theorem}\label{th:degenerate}
Suppose any non-degenerate set of $n$ disjoint segments can be cut into $\gamma(n)$ fragments
in $T(n)$ time such that the resulting set of fragments admits a depth order.
Then any set of $n$ disjoint segments can be cut into $O(\gamma(n)\log n)$ fragments
in $T(n)+O(n^2)$ time such that the resulting set of fragments admits a depth order.
\end{theorem}
\begin{proof}
The bound on the number of fragments immediately follows from the discussion above.
The overhead term in the running time is caused by the computation of the perturbed set~$S'$, which
can be done in $O(n^2)$ time if we compute the full arrangement in the projection.
\end{proof}

\paragraph{Degeneracies among triangles.}
Recall that the cuts we make on the triangles are induced by vertical planes, and that
a triangle becomes open where it is cut. When a triangle is completely contained
in the cutting plane, however, it is not well defined what happens. One option is to say
that the triangle completely disappears; another option it to
say that the triangle is not cut at all. Since vertical triangles can be ignored
in the Painter's Algorithm, we will simply assume that no triangle in $\T$ is vertical.
However, we can still have other degeneracies, such as edges of different
triangles being parallel or triples of projected edges being concurrent.
Fortunately, the fact that we do not need non-degeneracy assumptions for segments immediately implies
that we can handle such cases. Indeed, degeneracies are
not a problem for the hierarchical cuttings we use in Step~\ref{step1} of our procedure,
and in Step~\ref{step2} we only assumed non-degeneracy when computing the cut set $X(\sigma)$ for
the edge set~$\E(\sigma)$---and Theorem~\ref{th:degenerate} implies we can get rid of
the non-degeneracy assumptions in this step. Note that the $O(n^2)$ overhead term in
Theorem~\ref{th:degenerate} is subsumed by the time needed to apply the algorithm of Aronov~\etal~\cite{abgm-ccrs-08}.

%
%

\section{Concluding remarks}
\label{se:concl}
We proved that any set of $n$ disjoint triangles in $\Reals^3$ can be cut into
$O(n^{7/4}\polylog n)$ triangular fragments that admit a depth order, thus providing
the first subquadratic bound for this important setting of the problem.  We also proved
a refined bound that depends on the number of intersections of the triangle edges
in the projection, and generalized the result to $xy$-monotone surface patches.
The main open problem is to tighten the gap between our bound and the $\Omega(n^{3/2})$
lower bound on the worst-case number of fragments needed: is it possible to cut
any set of triangles into roughly $\Omega(n^{3/2})$ triangular fragments that admit a depth
order, or is this only possible by using curved cuts? One would expect the former,
but curved cuts seem unavoidable in the approach of Aronov, Miller and Sharir~\cite{ams-adcat-17}
and it seems very hard to push our approach to obtain any $o(n^{7/4})$ bound.

\bibliographystyle{plain}

\begin{thebibliography}{99}
\newcommand{\cgta}{\emph{Comput.\ Geom.\ Theory Appl.}\xspace}
\newcommand{\dcg}{\emph{Discr.\ Comput.\ Geom.}\xspace}
\newcommand{\ijcga}{\emph{Int.\ J.\ Comput.\ Geom.\ Appl.}\xspace}
\newcommand{\alg}{\emph{Algorithmica}\xspace}
\newcommand{\sicomp}{\emph{SIAM J.\ Comput.}\xspace}
\newcommand{\jalg}{\emph{J.\ Alg.}\xspace}
\newcommand{\jda}{\emph{J.\ Discr.\ Alg.}\xspace}
\newcommand{\ipl}{\emph{Inf.\ Proc.\ Lett.}\xspace}
\newcommand{\jacm}{\emph{J.~ACM}\xspace}

\newcommand{\socg}[1]{In \emph{Proc.\ #1 Int. Symp.\ Comput.\ Geom.\ (SoCG)}\xspace}
\newcommand{\soda}[1]{In \emph{Proc.\ #1 ACM-SIAM Symp.\ Discr.\ Alg.\ (SODA)}\xspace}
\newcommand{\stoc}[1]{In \emph{Proc.\ #1 ACM Symp.\ Theory Comp.\ (STOC)}\xspace}
\newcommand{\focs}[1]{In \emph{Proc.\ #1 IEEE Symp.\ Found. Comput. Sci.\ (FOCS)}\xspace}
\newcommand{\esa}[1]{In \emph{Proc.\ #1 Europ.\ Symp.\ Alg (ESA)}\xspace}
\newcommand{\swat}[1]{In \emph{Proc.\ #1 Scand.\ Workshop Alg.\ Theory (SWAT)}\xspace}
\newcommand{\isaac}[1]{In \emph{Proc.\ #1 Int.\ Symp.\ Alg.\ Comput. (ISAAC)}\xspace}

\bibitem{agmv-bspfr-00}
P.K.~Agarwal, E.F.~Grove, T.M.~Murali, and J.S.~Vitter.
Binary space partitions for fat rectangles.
\sicomp 29:1422--1448 (2000).

\bibitem{aks-rsisf-95}
P.K.~Agarwal, M.J.~Katz, and M.~Sharir.
Computing depth orders for fat objects and related problems.
\cgta 5:187--206 (1995).

\bibitem{abgm-ccrs-08}
B.~Aronov, M.~de Berg, C.~Gray, and E.~Mumford.
Cutting cycles of rods in space: Hardness results and approximation algorithms.
\soda{19th}, pages 1241--1248, 2008.

\bibitem{as-atbedc-16}
B.~Aronov and M.~Sharir.
Almost tight bounds for eliminating depth cycles in three dimensions.
\stoc{48th}, pages 1--8, 2016.

\bibitem{aks-ctcls-05}
B.~Aronov, V.~Koltun, and M.~Sharir.
Cutting triangular cycles of lines in space.
\dcg 33:231--247 (2005).

\bibitem{ams-adcat-17}
B.~Aronov, E.Y.~Miller, and M.~Sharir.
\emph{Eliminating depth cycles among triangles in three dimensions.}
\soda{28th}, pages 2476--2494, 2017.

\bibitem{b-rsdoh-93}
M.~de~Berg.
\newblock \emph{Ray Shooting, Depth Orders and Hidden Surface Removal}.
\newblock Springer-Verlag New York, LNCS 703, 1993.

\bibitem{b-lsbsp-00}
M. de Berg.
Linear size binary space partitions for uncluttered scenes.
\alg 28:353--366 (2000).

\bibitem{bcko-cgaa-08}
M. de Berg, O. Cheong, M. van Kreveld, and M. Overmars.
\emph{Computational Geometry: Algorithms and Applications.}
Springer, 2008.

\bibitem{bg-vrscd-08}
M.~de Berg and C.~Gray.
Vertical ray shooting and computing depth orders for fat objects.
\sicomp 38:257--275 (2008).

\bibitem{bos-cvdo-94}
M.~de~Berg, M.~Overmars, and O.~Schwarzkopf.
Computing and verifying depth orders.
\sicomp 23:437--446 (1994).

\bibitem{bs-ca-95}
M.~de~Berg and O.~Schwarzkopf.
Cuttings and applications.
\ijcga 5:43--55 (1995).

\bibitem{c-cpp-86}
B.~Chazelle.
Convex partitions of polyhedra: A lower bound and worst-case optimal algorithm.
\sicomp 13:488--507 (1984).

\bibitem{c-chdc-93}
B.~Chazelle.
Cutting hyperplanes for divide and conquer.
\dcg 9:145--158 (1993).

\bibitem{cegpsss-ccclr-92}
B.~Chazelle,  H.~Edelsbrunner,  L. J.~Guibas,  R.~Pollack,  R.~Seidel,  M.~Sharir,  and J.~Snoeyink. Counting and cutting cycles of lines and rods in space.
\cgta 1:305--323 (1992). (A preliminary version appeared in
\focs{31st}, pages 242--251, 1991.)

\bibitem{g-ppsv-15}
L. Guth.
Polynomial partitioning for a set of varieties.
\emph{Math. Proc. Cambridge Phil. Soc.} 159:459--469 (2015).

\bibitem{hs-oplpa-01}
S.~Har-Peled and M.~Sharir.
Online point location in planar arrangements and its applications.
\dcg~26: 19--40 (2001).

\bibitem{m-rsehc-93}
J.~Matou\v{s}ek.
Range searching with efficient hierarchical cuttings.
\dcg 10: 157--182 (1993).

\bibitem{py-ebsph-90}
M.~S.~Paterson and F.~F.~Yao.
Efficient binary space partitions for hidden-surface removal and solid modeling.
\dcg 5:485--503 (1990).

\bibitem{py-obspo-92}
M.~S.~Paterson and F.~F.~Yao.
Optimal binary space partitions for orthogonal objects.
\jalg~13:99--113 (1992).

\bibitem{s-ccoris-98}
A.~Solan.
Cutting cycles of rods in space.
\socg{14th}, pages 135--142, 1998.

\bibitem{t-bspaafr-08}
C.D.~T\'{o}th.
Binary space partitions for axis-aligned fat rectangles.
\sicomp~38: 429--447 (2008).

\bibitem{wl-grma-94}
R.H.~Wilson and J.-C.~Latombe.
Geometric reasoning about mechanical assembly.
\emph{Artificial Intelligence}~71:371--396 (1994).

\end{thebibliography}


\end{document}